\documentclass[11pt,letterpaper]{article}

\usepackage{caption}
\usepackage{amsfonts}
\usepackage{amsmath}
\usepackage{amssymb}
\usepackage{mathrsfs}
\usepackage{eufrak}
\usepackage{bussproofs}
\usepackage{tikz}
\usepackage{graphicx}
\usepackage{epstopdf}
\usepackage{url}
\usepackage{times}

\newif\ifdraft  

\ifdraft
\newcommand{\REMARK}[2]{{\bf [*[#1:~#2]*]}}
\newcommand{\erich}[1]{\REMARK{ERICH}{#1}}
\newcommand{\val}[1]{\REMARK{VAL}{#1}}
\newcommand{\finish}[1]{\REMARK{FINISH}{#1}}
\else
\newcommand{\REMARK}[2]{}
\newcommand{\erich}[1]{}
\newcommand{\val}[1]{}
\newcommand{\finish}[1]{}
\fi
\newcommand{\eat}[1]{}

\newtheorem{theorem}{Theorem}

\newtheorem{proposition}[theorem]{Proposition}
\newtheorem{corollary}[theorem]{Corollary}
\newtheorem{definition}[theorem]{Definition}
\newtheorem{example}[theorem]{Example}

\newenvironment{proof}{\textbf{Proof:}}{\hfill$\blacksquare$\\}

\renewcommand{\phi}{\varphi}

\newcommand{\Vars}{\mathsf{Vars}}
\newcommand{\eop}{\;\mathsf{op}\;}
\newcommand{\barx}{\mathbf{x}}
\newcommand{\bara}{\mathbf{a}}

\newcommand{\gothA}{\mathfrak{A}}
\newcommand{\gothE}{\mathfrak{E}}
\newcommand{\gothF}{\mathfrak{F}}
\newcommand{\gothG}{\mathfrak{G}}
\newcommand{\gothH}{\mathfrak{H}}
\newcommand{\gothp}{\mathfrak{p}}
\newcommand{\gothq}{\mathfrak{q}}

\newcommand{\scrV}{\mathscr{V}}

\newcommand{\FactsA}{\mathsf{Facts}_A}
\newcommand{\LitA}{\mathsf{Lit}_A}
\newcommand{\NegFactsA}{\mathsf{NegFacts}_A}

\newcommand{\nnf}{\mathsf{nnf}}

\newcommand{\LitV}{\mathsf{Lit}_V}
\newcommand{\piA}{\pi_{_\gothA}}
\newcommand{\pisA}{\pi_{_{\#\gothA}}}
\newcommand{\Api}{\gothA_\pi}
\newcommand{\Abeta}{\gothA_\beta}
\newcommand{\Modpi}{\mathsf{Mod}_{\pi}}

\newcommand{\refi}[2]{#1|_{_#2}}
\newcommand{\domi}{\mathsf{dominant}}
\newcommand{\noVdom}{\mathsf{noVdom}}
\newcommand{\notdom}{\mathsf{notdom}}
\newcommand{\denydom}{\mathsf{denydom}}
\newcommand{\nn}[1]{\bar{#1}}

\newcommand{\nnp}{\nn{p}}
\newcommand{\nnq}{\nn{q}}
\newcommand{\nnr}{\nn{r}}
\newcommand{\nns}{\nn{s}}
\newcommand{\nnt}{\nn{t}}
\newcommand{\nnX}{\nn{X}}
\newcommand{\nnx}{\nn{x}}
\newcommand{\nnz}{\nn{z}}
\newcommand{\Bool}{\mathbb{B}}
\newcommand{\Access}{\mathbb{A}}
\newcommand{\Nat}{\mathbb{N}}

\newcommand{\Trop}{\mathbb{T}}
\newcommand{\Vit}{\mathbb{V}}
\newcommand{\Fuzz}{\mathbb{F}}
\newcommand{\PosBool}{\mathsf{PosBool}}

\newcommand{\Pub}{\mathsf{P}}
\newcommand{\Cnf}{\mathsf{C}}
\newcommand{\Sec}{\mathsf{S}}
\newcommand{\Tsec}{\mathsf{T}}

\newcommand{\poly}[2]{#1\lbrack#2\rbrack}
\newcommand{\sem}[2]{\lbrack\!\lbrack#1\rbrack\!\rbrack_{#2}}

\newcommand{\ModDef}{model-defining}
\newcommand{\bfModDef}{\textbf{model-defining}}
\newcommand{\ModComp}{model-compatible}
\newcommand{\bfModComp}{\textbf{model-compatible}}

\newcounter{enum}
\newenvironment{enum}{\begin{list}{(\arabic{enum})}%
{\setlength{\labelwidth}{5mm}\setlength{\leftmargin}{10mm}%
\setlength{\itemindent}{0pt}\usecounter{enum}}}{\end{list}}

\title{Semiring Provenance for First-Order Model Checking}
\author{Erich Gr\"adel~~~~~~~~~~~~~~~~~~~~~~~~~~~~~Val Tannen\\
RWTH~Aachen~University~~~~~~~~~~~Univ.~of~Pennsylvania
}

\begin{document}
\maketitle

\begin{abstract}
Given a first-order sentence, a model-checking computation tests
whether the sentence holds true in a given finite structure. Data
provenance extracts from this computation an abstraction of the manner
in which its result depends on the data items that describe the
model. Previous work on provenance was, to a large extent, restricted to the negation-free
fragment of first-order logic and showed how provenance abstractions
can be usefully described as elements of commutative semirings --- most 
generally as multivariate polynomials with positive integer coefficients.

In this paper we introduce a novel approach to dealing with negation
and a corresponding commutative semiring of polynomials with dual
indeterminates. These polynomials are used to perform reverse provenance
analysis, i.e., finding models that satisfy various properties under
given provenance tracking assumptions.
\end{abstract}

\section{Introduction}

Semiring provenance was originally developed
for positive database query languages~\cite{pods/GreenKT07}.
From this baseline, we have recently started to investigate 
an approach to the provenance analysis of model checking for full first-order logic (FOL).
We propose a novel approach to dealing with negation in provenance
formulation and a corresponding commutative semiring of polynomials with dual
indeterminates. A preliminary account of this joint work
was given by the second author in~\cite{siglog/Tannen17}.

Data provenance is extremely useful in many computational disciplines.
Suppose that a computational process is applied to a complex input consisting
of multiple items. Provenance analysis allows us to understand how these
different input items affect the output of the process. It can be used to
answer questions of the following type:
\begin{enum}
\item Which ones of input items are actually used 
in the computation of the output? 
\item Can the same output be obtained from
different combinations of input items? 
\item In how many different ways
can the same output be computed? 
\end{enum}

As a consequence, provenance can be further
applied to issues such as deciding how much to trust the output,
assuming that we may trust some input items more than others, deciding
what clearance level is required for accessing the output, assuming that
we know the clearance levels for the input items, or, assuming that  one has to pay
for the input items, how to minimize the cost of obtaining the output.
More generally, \emph{reverse} provenance analysis
allows us to find input data (here first-order models) that satisfies 
various properties under given provenance tracking assumptions.
This is also closely related to reverse data 
management~\cite{pvldb/MeliouGS11,sigmod/MeliouS12}.

It turns out that the questions listed above, as well as several other
questions of interest,
can be answered for database transformations (queries and views) via
interpretations in commutative semirings. In past work, the semiring
provenance approach has been applied to query and view languages such
as the positive relational algebra~\cite{pods/GreenKT07,mst/Green11},
nested relations/complex values 
(objects)~\cite{pods/FosterGT08,birthday/Tannen13},
Datalog~\cite{pods/GreenKT07,icdt/DeutchMRT14},
XQuery (for unordered XML)~\cite{pods/FosterGT08}
full relational algebra 
(on $\mathbb{Z}$-annotated relations)~\cite{mst/GreenIT11},
SQL aggregates~\cite{pods/AmsterdamerDT11},
workflows with map-reduce modules~\cite{pvldb/AmsterdamerDDMST11},
and languages for data-centric (data-dependent) 
processes~\cite{vldb/DeutchMT15}. Moreover, the semiring approach has been 
successfully implemented in two software systems, \textsf{Orchestra}
\cite{vldb/GreenKIT07,sigmod/IvesGKTTTJP08,sigmod/KarvounarakisIT10} and 
\textsf{Propolis}~\cite{vldb/DeutchMT15}.

There exists a well-known tight connection between conjunctive queries
in databases and constraint satisfaction problems in
AI~\cite{jcss/KolaitisV00}. In this light, and despite a number of
technical differences, there exists an interesting connection (that
needs more exploration) between the semiring provenance framework
applied to conjunctive queries and the semiring framework for
\emph{soft constraint satisfaction}~\cite{jacm/BistarelliMR97,sp/Bistarelli04}.

The reader may have noticed that the bulk of the work on provenance
for database transformations was concerned with \emph{positive} query
languages. Indeed, trying to add to the commutative semiring structure
operations that capture difference of relations has led to
interesting and algebraically challenging, but divergent
approaches~\cite{japll/GeertsP10,mst/GreenIT11,pods/AmsterdamerDT11,tapp/AmsterdamerDT11,jacm/GeertsUKFC16}. In particular there is no separate account
of tracking \emph{negative} information, an aspect that we hope to remedy
here.

\subsection{Provenance Semantics}
\label{subsec:sem}
We shall consider certain non-standard semantics for FOL
that will help us to understand how a sentence $\phi$ ends up being true
in a finite structure $\gothA$, i.e., 
whether $\gothA\models\phi$ holds or not (we call this  
\emph{provenance in model checking}).
The non-standard
semantics that we champion involves various \emph{commutative
  semirings}.  Here we strive to justify this choice.

First of all, the standard semantics for first-order logic maps formulae
to truth values in $\Bool=\{\bot,\top\}$, which 
form a commutative semiring with respect 
to the operations of disjunction and conjunction, with units $\bot$
and $\top$.

Second, in a provenance semantics we want to understand the
connections between the facts (positive or negative) that are embodied
in a model $\gothA$ and their use in a justification that
$\gothA\models\phi$. Since the model is finite, we can think of such a
justification as an alternating disjunction-conjunction
\emph{proof tree} (an example appears in~\ref{subsec:ex}). 
In any case, these justifications are definitely
not proofs in some axiomatization of FOL. If we had a provenance semantics
for model checking, it would, in particular, help us to count proof trees.
This particular case suffices to suggest the semiring structure as well
as some ways in which such non-standard semantics can be quite different
from the standard one.

Notice that a semiring semantics refines the classical Boolean semantics,
and formulae that are classically equivalent may become 
non-equivalent under a semantics that counts proof trees.
Indeed, already a sentence $\phi\vee\phi$ has in general
more proof trees than $\phi$.
We further illustrate with the failure
of some of the usual logical equivalences
invoked in transforming sentences to \emph{prenex form}.

Let $\rho\equiv(\forall x\, \phi) \wedge \psi$ and $\sigma\equiv
\forall x\, (\phi \wedge \psi)$.  Every proof tree of $\rho$ can be
transformed into a proof tree of $\sigma$ by making copies of the
subtree rooted at $\psi$.  However, when $\psi$ has two or more
distinct proof trees we see that $\sigma$ can have strictly more proof
trees than $\rho$. Similarly we can argue that
$\forall x\, (\phi \vee \psi)$ can have strictly more proof trees
than $(\forall x\, \phi) \vee \psi$.

Now consider $\rho\equiv(\exists x\, \phi) \vee \psi$
and $\sigma\equiv\exists x\, (\phi \vee \psi)$. Let's write
$\phi(x)$ to show occurrences of $x$ in $\phi$. For simplicity
suppose that the model has exactly two elements, $a$ and $b$, and that
each of $\phi(a)$, $\phi(b)$, and $\psi$ has exactly one proof tree.
Then, $\rho$ will have 3 proof trees but $\sigma$ will have 4.

Finally, we note that $(\exists x\, \phi) \wedge \psi)$
and $\exists x\, (\phi \wedge \psi)$ have exactly the same number
of proof trees and this reflects the fact that multiplication distributes 
over addition. 

For other sentences, we can see that the number-of-proof-trees
constitutes a non-standard semantics for FOL sentences constructed
using disjunction, conjunction, existentials and universals,
because, moreover, addition and multiplication are
associative and commutative.

This discussion provides some partial justification for considering
commutative semirings as semantic domains. The rest of the justification
will follow from the subsequent development.

\medskip\noindent{\bf Remark. }
Instead of thinking about proof trees for $\gothA\models\phi$, 
we could equivalently consider \emph{winning strategies} in $\mathcal{G}(\gothA,\phi)$,
the model checking game for $\gothA$ and $\phi$ (see e.g. \cite{AptGra11}).
We do not pursue this aspect in this paper, but we remark that a provenance analysis in commutative 
semirings can also be developed for more general models of finite and infinite games, 
beyond the acyclic and always terminating first-order  model-checking games.
Also beyond the applications to query evaluation and logic, a provenance analysis 
of games provides insights into more subtle game-theoretic questions than just who wins the game, 
concerning for instance  the number or costs of winning strategies, or issues like confidence and trust in game-theoretic settings.
This approach will be developed in more detail in a forthcoming paper.

\subsection{Intermezzo: Examples of Commutative Semirings}
\label{subsec:commutative-semirings}

\begin{definition} An algebraic structure 
$(K,+,\cdot,0,1)$, with $0\neq1$, is a \emph{semiring} when $(K,+,0)$
is a commutative monoid, $(K,\cdot,1)$ is a monoid, $\cdot$
distributes over $+$ and $0\cdot a=a\cdot0=0$. The semiring
is \emph{commutative} when $\cdot$ is commutative, and it 
is \emph{idempotent} when $+$ is idempotent.
\end{definition}

Any distributive lattice is an idempotent commutative semiring.
Here are some commutative semirings of interest to us: 
\begin{enumerate}
\item
The Boolean semiring $\Bool=(\Bool,\vee,\wedge,\bot,\top)$ is the standard habitat of 
logical truth. It is a distributive lattice.
\item
$\Nat=(\Nat,+,\cdot,0,1)$ is used for \emph{bag semantics} in databases
and we use it here for counting proof trees.
\item
$\Trop=(\mathbb{R}_{+}^{\infty},\min,+,\infty,0)$ 
is called the \emph{tropical} semiring and is idempotent but
not a distributive lattice. Its elements and operations appear in
\emph{min-cost} interpretations (e.g., shortest paths) and it plays
a surprising role in connecting certain dynamic programming algorithms
in statistics with certain methods of algebraic 
geometry~\cite{book/PachterS05} (see also next item). 
\item
$\Vit=([0,1],\max,\cdot,0,1)$
is called the \emph{Viterbi} semiring and is isomorhic to $\Trop$
via $x\mapsto e^{-x}$ and $y\mapsto -\ln y$. When interpreted as 
probabilities, its elements and operations appear in 
\emph{statistical model} interpretations 
(e.g., maximum probability trajectories in Hidden Markov Models).
We will think of the elements of $\Vit$ as \emph{confidence scores}.
\item
$\Fuzz=([0,1],\max,\min,0,1)$, is called the \emph{fuzzy} semiring. 
It is a distributive lattice.
\item
$\Access=(\{\Pub<\Cnf<\Sec<\Tsec<0\},\min,\max,0,\Pub)$
is the \emph{access control} semiring, where
$\Pub$ is ``public'',  $\Cnf$ is ``confidential'', $\Sec$ is ``secret'', $\Tsec$ is ``top secret'', and
$0$ is ``so secret that nobody can access it!''.
This is a distributive lattice (beware! the lattice order is the opposite
of the one we used in the definition).
\item For any set $X$, the semiring $\poly{\Nat}{X}=(\poly{\Nat}{X},+,\cdot,0,1)$
consist of the multivariate polynomials in indeterminates from $X$
and with coefficients from $\Nat$. This is the commutative
semiring freely generated by the set $X$. It's used for a general
form of provenance.
\item
$\PosBool(X)=(\PosBool(X),\vee,\wedge,\bot,\top)$ is the semiring
whose elements are classes of equivalent positive (monotone)
boolean expressions with boolean variables from $X$ (its elements
are in bijection with the positive boolean expressions 
in irredundant disjunctive normal form).
This is the distributive lattice freely generated by the set $X$.
It is also used for provenance, e.g., in probabilistic databases.
\end{enumerate}

\section{First-Order Logic Interpreted in Commutative Semirings}

We are interested in the provenance analysis of the model checking
computation of first-order sentences. Such a computation is nicely and
\emph{declaratively} driven by the structure of the sentence, and
thus amounts to a non-standard semantics for FOL.  In its simplest
form model checking takes as input a finite structure and
the input items are the various facts (positive or negative) which
hold in the model. We have found however that it pays to take a more
general approach and specify not a structure but just its (finite)
universe. This way we can track the use of positive and negative facts
in checking a sentence under multiple possible models on that universe.
This allows a certain amount of \emph{reverse analysis}: finding models
that satisfy useful constraints.

\subsection{$K$-Interpretations}
Consider a finite relational vocabulary: $\scrV=\{R,S,\ldots\}$.
From this vocabulary and a finite \emph{non-empty} universe $A$ of \emph{ground
values} we construct the set $\FactsA$ of all ground relational
atoms (facts) $R(\bara)$, the set $\NegFactsA$ of all negated
facts $\neg R(\bara)$ and thus the set
$\LitA=\FactsA\cup\NegFactsA$ of all \emph{literals},
positive and negative facts, over $\scrV$ and $A$.
By convention we will identify $\neg\neg R(\bara)\equiv R(\bara)$)
so the negation of a literal is again a literal.

Any finite  structure $\gothA = (A,R^{\gothA},S^{\gothA},\ldots)$ 
with universe $A$ makes some of these
literals true and the remaining ones false. Note, however, that
much of the development does not assume a specific model, and this 
can be usefully exploited.

Let $(K,+,\cdot,0,1)$ be a commutative semiring. Very roughly
speaking, $0\in K$ is intended to interpret false assertions,
while an element $a\neq0$ in $K$ provides a ``nuanced'' interpretation
for true assertions (call them ``$a$-true'').

Next, $K$-interpretations will map literals to elements of $K$
and are then extended to all formulae.
Disjunction and existential quantification are interpreted by the
addition operation of $K$.
Conjunction and universal quantification are interpreted by the
multiplication operation of $K$. 
For quantifiers, the finiteness of the universe
$A$ of ground values will be essential. 
For negation we use the well-known syntactic transformation to 
\emph{negation normal form (NNF)}, denoted $\psi\mapsto\nnf(\psi)$.
Note that $\nnf(\psi)$ is a formula constructed from 
literals (positive and negative facts) and equality/inequality atoms 
using just $\wedge,\vee,\exists,\forall$. 

\begin{definition}
\label{def:sem}
A $K$\textbf{-interpretation} is a mapping $\pi:\LitA\rightarrow K$.
This is extended to FO formulae given valuations 
$\nu:\Vars\rightarrow A$: 
$$
\begin{array}{rcl@{\hspace*{10mm}}rcl}
\pi\sem{R(\barx)}{\nu}  & = &  \pi(R(\nu(\barx))
&
\pi\sem{\neg R(\barx)}{\nu}  & = &  \pi(\neg R(\nu(\barx))\\
\\
\pi\sem{x\eop y}{\nu} & = & \mathrm{if~}\nu(x)\eop\nu(y)
                                 \mathrm{~then~} 1 \mathrm{~else~}0
&
\pi\sem{\phi\wedge\psi}{\nu} & = & \pi\sem{\phi}{\nu}~\cdot~
                                   \pi\sem{\psi}{\nu}\\
\\
\pi\sem{\phi\vee\psi}{\nu} & = & \pi\sem{\phi}{\nu}~+~
                                   \pi\sem{\psi}{\nu}
&
\pi\sem{\exists x\,\phi}{\nu} & = &  \sum_{a\in A}\pi\sem{\phi}
                                        {\nu[x\mapsto a]}\\
\\
\pi\sem{\forall x\,\phi}{\nu} & = & \prod_{a\in A}\pi\sem{\phi}
                                       {\nu[x\mapsto a]}
&
\pi\sem{\neg\phi}{\nu} &=& \pi\sem{\nnf(\neg\phi)}{\nu}
\end{array}
$$
\end{definition}

The symbol $\eop$ stands for either $=$ or $\neq$. 
As you can see from the definition,
the equality and inequality atoms are interpreted in $K$ as $0$ or $1$, i.e.,
their provenance is not tracked. One could give a similar treatment to other
such relations with ``fixed'' meaning, e.g., assuming an ordering on $A$, 
however, we omit this here.

As intended, it suffices to consider formulae in NNF:

\begin{proposition} 
\label{prop:NNF}
\hspace*{5mm}
$\pi\sem{\phi}{\nu} ~=~ \pi\sem{\nnf(\phi)}{\nu}$ 
\end{proposition}

\begin{corollary}~~~
$\begin{array}[t]{rcl}
\pi\sem{\neg(\phi\wedge\psi)}{\nu} & = &
\pi\sem{\neg\phi\vee\neg\psi)}{\nu}\\
\pi\sem{\neg(\phi\vee\psi)}{\nu} & = & 
\pi\sem{\neg\phi\wedge\neg\psi)}{\nu}\\
\pi\sem{\neg(\forall x\,\phi)}{\nu} & = &
\pi\sem{\exists x\,\neg\phi}{\nu}\\
\pi\sem{\neg(\exists x\,\phi)}{\nu} & = &
\pi\sem{\forall x\,\neg\phi}{\nu}\\
\pi\sem{\neg\neg\phi}{\nu} & = & \pi\sem{\phi}{\nu}
\end{array}$
\end{corollary}

A useful consequence of Proposition~\ref{prop:NNF} is that we can prove
further results by induction on formulas in NNF, and hence avoid the
negation connective.
When $\phi$ is a sentence we write just $\pi\sem{\phi}{}$.

\begin{proposition}[Fundamental Property]
\label{prop:hom}
Let $h:K_1\rightarrow K_2$ be a semiring homomorphism
and let $\pi_1:\LitA\rightarrow K_1$ and
$\pi_2:\LitA\rightarrow K_2$ be interpretations such that
$h\circ\pi_1=\pi_2$. Then, for any FOL sentence $\phi$ we have
$h(\pi_1\sem{\phi}{}) ~=~ \pi_2\sem{\phi}{}$. As diagrams

~~~~~~~~~~~~~~~~~~~~~
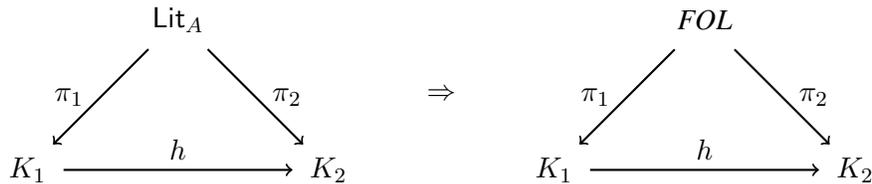
\begin{tikzpicture}
\begin{scope}
  \node[shape=circle] (A) at (0,0) {$K_1$};
  \node[shape=circle] (B) at (2,2) {$\LitA$};
  \node[shape=circle] (C) at (4,0) {$K_2$};

 \draw[->,thick] 
      (B) edge node[left] {$\pi_1\:$} (A);
 \draw[->,thick] 
      (B) edge node[right] {$\:\pi_2$} (C);
 \draw[->,thick] (A) edge node[above] {$h$} (C);

\node[shape=circle] (D) at (5.5,1) {$\Rightarrow$};

  \node[shape=circle] (A) at (7,0) {$K_1$};
  \node[shape=circle] (B) at (9,2) {FOL};
  \node[shape=circle] (C) at (11,0) {$K_2$};

 \draw[->,thick] 
      (B) edge node[left] {$\pi_1\:$} (A);
 \draw[->,thick] 
      (B) edge node[right] {$\:\pi_2$} (C);
 \draw[->,thick] (A) edge node[above] {$h$} (C);
 \end{scope}
\end{tikzpicture}
\end{proposition}

\begin{proof}
By Proposition~\ref{prop:NNF} the proof can proceed by induction
on formulae in NNF. For example
$h(\pi_1\sem{\phi\wedge\psi}{\nu}) =
h(\pi_1\sem{\phi}{\nu} \cdot_{_1} \pi_1\sem{\psi}{\nu}) = 
h(\pi_1\sem{\phi}{\nu}) \cdot_{_2} h(\pi_1\sem{\psi}{\nu}) = 
\pi_2\sem{\phi}{\nu} \cdot_{_2} \pi_2\sem{\psi}{\nu} = 
\pi_2\sem{\phi\wedge\psi}{\nu}$.
\end{proof}

The somewhat bombastic name ``fundamental property'' is motivated by
two observations. First, the property checks that the
definition of our semantics is nicely compositional. Second,
the property plays a central role in a strategy that we have widely
applied with query languages in databases: compute provenance
as generally as (computationally) feasible, then specialize via
homomorphisms to
coarser-grain provenance, or to specific domains, e.g., count, trust, cost
or access control.

\subsection{Intermezzo: Positive Semirings}
We say that a semiring $K$ has \emph{divisors of 0} if there
exist $a,b\in K$ such that $a\neq0$, $b\neq0$ but $ab=0$. None
of the semirings described in Sect.~\ref{subsec:commutative-semirings}
has divisors of 0. The classical examples of such are rings that
are not integral domains, e.g., $\mathbb{Z}_6$, as well as boolean algebras.

A semiring $K$ is \emph{+-positive} if $a+b=0$ implies $a=0$ and
$b=0$. Rings, e.g., $\mathbb{Z}$, or the boolean ring $\mathbb{Z}_2$,
are not $+$-positive. Finally, a semiring is (simply)
\emph{positive}~\cite{book/Eilenberg74} if it is $+$-positive and has no
divisors of $0$. All the semirings described 
in Sect.~\ref{subsec:commutative-semirings} are positive.

\begin{proposition}
\label{prop:positive-semiring}
A semiring $K$ is positive if, and only if,
$~~\dagger_{_K}:K\rightarrow\Bool~~$ defined by

\hspace*{1cm}
\begin{minipage}{0.5\textwidth}
$$
\dagger_{_K}(a)= \begin{cases}
              \top & \mathrm{if~} a\neq0\\
              \bot & \mathrm{if~} a=0
              \end{cases}
$$
\end{minipage}
\begin{minipage}{0.4\textwidth}
is a homomorphism.
\end{minipage}
\end{proposition}

\subsection{Sanity Checks}

Let~~$\gothA = (A,R^{\gothA},S^{\gothA},\ldots)$~~be a 
(finite) $\scrV$-model. 

The \textbf{canonical truth interpretation} for $\gothA$
is, of course, ~$\piA:\LitA\rightarrow\Bool$ where
$$
\piA(L) ~=~ \begin{cases}
            \top & \mathrm{if~} \gothA\models L\\
            \bot & \mathrm{otherwise}
            \end{cases}
$$

Earlier we have discussed ``number of proof trees''
as a non-standard semantics for FOL model-checking.
This is also captured by interpretations in a semiring.

The \textbf{canonical counting interpretation} for $\gothA$
is ~~$\pisA : \LitA\rightarrow\Nat$ where
$$
\pisA(L) ~=~ \begin{cases}
              1 & \mathrm{if~} \gothA\models L\\
              0 & \mathrm{otherwise}
              \end{cases}
$$

\begin{proposition}[sanity checks]
\label{prop:sanity}
For any FOL sentence $\phi$ we have
$\gothA\models\phi$ if, and only if,
$\piA\sem{\phi}{}=\top$.
Moreover, $\pisA\sem{\phi}{}$ is the number of proof trees that
witness $\gothA\models\varphi$.
\end{proposition}

Now, let $K$ be a commutative semiring, and let $\pi$
be a $K$-interpretation.  As we have indicated, for a sentence
$\phi$ we intend to interpret $\pi\sem{\phi}{}=0$ as ``$\phi$ is false
in $K$'', while $\pi\sem{\phi}{}=k\neq0$ is interpreted as ``$\phi$ is
$k$-true in $K$'', i.e., as offering ``shades of truth''. We examine
how this meshes with standard logical truth in a model.

\begin{definition}
\label{def:int-model-defining}
A $K$-interpretation $\pi:\LitA\rightarrow K$ is 
\bfModDef\
when, for each fact, one of $\pi(R(\bara))$ and $\pi(\neg R(\bara))$
is $0$ and the other one is $\neq0$.

Indeed, every \ModDef\ interpretation $\pi$ uniquely defines a 
$\scrV$-model $\Api$ with universe $A$ such that for any literal
$L$ we have $\Api\models L$ if, and only if, $\pi(L)\neq0$.
\end{definition}

Both $\piA$ and $\pisA$ shown above are \ModDef\ and the model they
define is $\gothA$. If $K$ is not $\Bool$ then several
\ModDef\ interpretations may define the same model. It is also
clear that any finite model can be defined by such an interpretation,
for any $K$.

\begin{proposition}[another sanity check]
\label{prop:strong}
Let $K$ be \emph{positive}, 
and let $\pi$ be a \ModDef\ $K$-interpretation.
Then for any FOL sentence
$$
\Api\models\phi ~~\Leftrightarrow~~\pi\sem{\phi}{}\neq0
$$ 
\end{proposition}

\begin{proof}
By Proposition~\ref{prop:positive-semiring},
since $K$ is positive, $\dagger_{_K}$ is a homomorphism.
Since $\pi$ is \ModDef\ let $\gothA$ be the model defined by $\pi$.
Clearly, $\dagger_{_K}\circ\pi$ is the canonical truth interpretation $\piA$. 
Applying Proposition~\ref{prop:hom}
we get $\dagger_{_K}(\pi\sem{\phi}{}) = \piA\sem{\phi}{}$. 
Now the result follows from Proposition~\ref{prop:sanity}.
\end{proof}

In fact, we can refine the previous proposition as follows.

\begin{proposition}[refinement of Proposition~\ref{prop:strong}]~~\\
\label{prop:refi}
\vspace*{-8mm}
\begin{enumerate}
\item[\rm (a)]
For any semiring $K$ (positive or not!),
for any \ModDef\ $K$-interpretation $\pi$, and 
for any FOL sentence $\phi$ we have
$$
\pi\sem{\phi}{}\neq0 ~~\Rightarrow~~ \Api\models\phi. 
$$ 
\item[\rm (b)]
Moreover, a semiring $K$ is positive if, and only if,
for any \ModDef\ $K$-interpretation $\pi$
and any FOL sentence $\phi$ we have
$$
\Api\models\phi ~~\Rightarrow~~\pi\sem{\phi}{}\neq0.
$$ 
\end{enumerate}
\end{proposition}
\begin{proof} 
Part (a) of the proposition is by induction on $\phi$.

The left to right implication in part (b) follows 
from Proposition~\ref{prop:strong}. For the right to left
implication we first prove that $K$ has no divisors of $0$.
Suppose that $a,b\in K$ are such that $a\neq0$, $b\neq0$ but
$ab=0$. Consider $A=\{\{c_1,c_2\}$ and the \ModDef\
interpretation defined by $\pi(\neg R(c_1))=\pi(\neg R(c_2))=0$,
$\pi(R(c_1))=a$, $\pi(R(c_2))=b$ as well as the sentence
$\phi=R(c_1)\wedge R(c_2)$. We have $\Api\models\phi$ 
hence $\pi\sem{\phi}{}\neq0$, contradiction.

Next we prove that $K$ is $+$-positive. Let $a,b\in K$ be such 
that $a\neq0$ and $b\neq0$. Consider the same interpretation
$\pi$ as above, with the sentence 
$\psi=R(c_1)\vee R(c_2)$. We have $\Api\models\psi$ hence 
$0\neq\pi\sem{\psi}{}=a+b$.
\end{proof}

\subsection{``Consistency'' and ``completeness'' for $K$-interpretations}

In the study of provenance we shall also have occasion to consider
interpretations that do not correspond to a single specific model
(as formalized in Definition~\ref{def:int-model-defining}). 
Additional issues arise for such interpretations.

An interpretation in which both $\pi\sem{\phi}{}\neq0$ and
$\pi\sem{\neg\phi}{}\neq0$ for some sentence $\phi$ is seemingly
``inconsistent''. On the other hand, an interpretation
in which both $\pi\sem{\phi}{}=0$ and
$\pi\sem{\neg\phi}{}=0$ for some sentence $\phi$ seems to to be
``incomplete''. 
\footnote{The same terminology is used for logical theories.}
Of course, neither of these situations arises for a \ModDef\
$K$-interpretation when $K$ is positive (by Proposition~\ref{prop:strong}).
We analyze each of these issues in turn for general
interpretations.

First we note that we have the following:
\begin{proposition}
\label{prop:K-consistency-one}
Let $\pi:\LitA\rightarrow K$ be a $K$-interpretation.
If for every $L\in\LitA$ at least one of $\pi(L)$
and $\pi(\neg L)$ is $0$
then there exists no sentence $\;\phi\;$ for which both 
$\pi\sem{\phi}{}\neq0$ and $\pi\sem{\neg\phi}{}\neq0$.
\end{proposition}

Observe that if at least one of $\pi\sem{\phi}{}$ or
$\pi\sem{\neg\phi}{}$ is $0$ then 
$\pi\sem{\phi}{}\cdot\pi\sem{\neg\phi}{}=0$. If $K$ has no
divisors of 0 the converse holds as well. Although the examples
described in~\ref{subsec:commutative-semirings} 
are positive semirings, we are about to introduce,
in~\ref{subsec:dual-ind-poly}, a semiring for FOL provenance that
\emph{does} have divisors of $0$. For this reason we note also the
following:
\begin{proposition}
\label{prop:K-consistency-two}
Let $\pi:\LitA\rightarrow K$ be a $K$-interpretation.
If for every $L\in\LitA$ we have $\pi(L)\cdot\pi(\neg L)=0$
then for any sentence $\phi$ we have
$\pi\sem{\phi}{}\cdot\pi\sem{\neg\phi}{}=0$.
\end{proposition}

Propositions~\ref{prop:K-consistency-one} and~\ref{prop:K-consistency-two}
hold in arbitrary $K$ and each supports a kind of ``consistency'', with
the two kinds coinciding when $K$ has no divisors of $0$.

Turning to ``completeness'', note that if both $\pi\sem{\phi}{}$ and
$\pi\sem{\neg\phi}{}$ are $0$ then
$\pi\sem{\phi}{}+\pi\sem{\neg\phi}{}=0$.  If $K$ is +-positive then
the converse holds as well.  However, for arbitrary $K$,
neither an analog of
Proposition~\ref{prop:K-consistency-one} nor one of
Proposition\ref{prop:K-consistency-two} holds. Indeed, let
$K=\mathbb{Z}_4$. Consider the vocabulary consisting of one unary
relation symbol $R$ and let $A=\{c_1,c_2\}$.  For the interpretation
given by $\pi(\neg R(c_1))=\pi(\neg R(c_2))=\pi(R(c_1))=\pi(R(c_2))=2$
and the sentence $\phi=R(c_1)\wedge R(c_2)$ we have
$\pi\sem{\phi}{}=\pi\sem{\neg\phi}{}=0$.

Instead, we have the following for positive semirings.

\begin{proposition}
\label{prop:K-completeness}
Assume that $K$ is positive.
Let $\pi:\LitA\rightarrow K$ be a $K$-interpretation.
If for every $L\in\LitA$ we have $\pi(L)\neq0$ or $\pi(\neg L)\neq0$
(equivalently, $\pi(L)+\pi(\neg L)\neq0$)
then for any sentence $\phi$ we have
$\pi\sem{\phi}{}\neq0$ or $\pi\sem{\neg\phi}{}\neq0$
(equivalently, $\pi\sem{\phi}{}+\pi\sem{\neg\phi}{}\neq0$).
\end{proposition}

\section{A Provenance Semiring for FOL} 
\label{sec:prov-semiring-FOL}

We have claimed  Sect.~\ref{subsec:commutative-semirings} that
$\poly{\Nat}{Y}$, the commutative semiring freely generated by a set
$Y$ is used for provenance tracking. The elements of $Y$ label the
information whose propagation we wish to capture in provenance. This
works fine for \emph{positive} database query
languages~\cite{pods/GreenKT07} but difference/negation cause problems.
Here we shall use a variation on the idea of polynomials
in order to deal with negated facts in provenance analysis.

We construct a semiring whose elements can be identified with certain
polynomials that describe the provenance of FOL model checking. The
main insight is the use of indeterminates in ``positive-negative pairs''.
We show that the resulting polynomials provide a nicely dual interpretation
for provenance that captures model-checking proofs.
We illustrate with a running example. 


\subsection{Dual-Indeterminate Polynomials}
\label{subsec:dual-ind-poly}

Let $X,\nnX$ be two disjoint sets together with a one-to-one
correspondence $X\longleftrightarrow\nnX$. We denote by $p\in X$ and
$\nnp\in\nnX$ two elements that are in this correspondence.  We refer
to the elements of $X\cup\nnX$ as \textbf{provenance tokens} as they
will be used to label/annotate some of the ``data'', i.e., literals
over some ground values, via the concept of $K$-interpretation that
we defined previously. Indeed, if, as before, we fix a finite 
non-empty set $A$ and consider $\LitA=\FactsA\cup\NegFactsA$ then
we shall use $X$ for $\FactsA$ and $\nnX$ for $\NegFactsA$. By
convention, if we annotate $R(\bara)$ with the ``positive'' token $p$ 
then the ``negative'' token $\nnp$ can only be
used to annotate $\neg R(\bara)$, and vice versa. 
We refer to $p$ and $\nnp$ as \emph{complementary} tokens.

Further, we denote by $\poly{\Nat}{X,\nnX}$ the quotient
of the semiring of polynomials $\poly{\Nat}{X\cup\nnX}$ by the
congruence generated by the equalities $p\cdot\nnp=0$ for
all $p\in X$.\footnote{This is the same as quotienting by the ideal generated 
by the polynomials $p\nnp$ for all $p\in X$.}
Observe that two polynomials $\gothp,\gothq\in\poly{\Nat}{X\cup\nnX}$
are congruent if, and only if, they become identical after deleting from each of them
the monomials that contain complementary tokens. Hence, the 
congruence classes in $\poly{\Nat}{X,\nnX}$ are in one-to-one correspondence
with the polynomials in $\poly{\Nat}{X\cup\nnX}$
such that none of their monomials contain complementary tokens.
We shall call these \textbf{dual-indeterminate polynomials} although we might
often omit ``-indeterminate'' just use ``dual polynomials''.

The following is the universality property of the semiring of dual
polynomials:
\begin{proposition}
\label{prop:prov-univ}
For any commutative semiring $K$ and for any
$f:X\cup\nnX\rightarrow K$ such that 
$\forall p\in X\,f(p)\cdot f(\nnp)=0$ there exists
a unique semiring homomorphism 
$h:\poly{\Nat}{X,\nnX}\rightarrow K$ 
such that $\forall x\in X\cup\nnX\;h(x)=f(x)$.
\end{proposition}

We note that $\poly{\Nat}{X,\nnX}$ is $+$-positive, but not positive,
since it has divisors of $0$. Examples:
$$
p\cdot\nnp~=~0,~~~~(p+\nnq)\nnp q ~=~0,~~~~(p\nnq+\nnp q)(pq+\nnp\nnq)~=~0.
$$
However, keeping both $p$ and $\nnp$ around
and even using them in certain ``inconsistent''
$\poly{\Nat}{X, \nnX}$-interpretations 
can be very useful in provenance analysis, as we shall see 
in Sect.~\ref{subsec:rev-ex}.

\begin{definition}
A \textbf{provenance-tracking} interpretation is a 
$\poly{\Nat}{X,\nnX}$-interpretation 
$\pi:\LitA\rightarrow\poly{\Nat}{X,\nnX}$ such that
$\pi(\FactsA)\subseteq X\cup\{0,1\}$
and $\pi(\NegFactsA)\subseteq\nnX\cup\{0,1\}$.
\end{definition}

The idea is that if $\pi$ annotates a positive or negative fact with
a token, then we wish to track that fact through the model-checking
computation. On the other hand annotating with $0$ or $1$ is done when
we do not track the fact, yet we need to recall whether it holds or not
in the model.

\subsection{An Example and a Characterization}
\label{subsec:ex}

The vocabulary of directed graphs consists one binary predicate $E$
denoting directed edges. 
Consider, over this vocabulary, the following formula and sentence
$$
\domi(x) ~\equiv~
\forall y\: \bigl( x=y \vee (E(x,y)\wedge\neg E(y,x))\bigr),
~~~~~~~~~~~~~~~~
\phi~\equiv~ \forall x\,\neg\domi(x).
$$
$\domi(x)$ says that in a digraph with edge relation $E$ 
the vertex $x$ is ``dominant'' while $\phi$ says
that the digraph does not have a dominant vertex.

Consider also the digraph $\gothG$ depicted in Figure~\ref{fig:modelG}
with vertices $a,b,c$. The edges of the digraph are the
solid arrows and we wish to track their presence through
model-checking.  The dashed arrows corresponds to absent edges, whose
absence, however, we also wish to track. 
We do this with the 
provenance-tracking $\poly{\Nat}{X,\nnX}$-interpretation
$\beta:\LitV\rightarrow X\cup\nnX\cup\{0,1\}$ defined by\\

\begin{minipage}[t]{7cm}
\hspace*{1cm}
$
\beta(L) ~=~ \begin{cases}
           p  & \mathrm{if~} L=E(a,b)\\
           0  & \mathrm{if~} L=\neg E(a,b)\\
           q  & \mathrm{if~} L=E(b,c)\\
           0  & \mathrm{if~} L=\neg E(b,c)\\
           0  & \mathrm{if~} L=E(a,c) \\
        \nnr  & \mathrm{if~} L=\neg E(a,c)
            \end{cases}
$
\end{minipage}
\begin{minipage}[t]{7cm}
$        =~\begin{cases}
           0  & \mathrm{if~} L= E(c,b)\\
        \nns  & \mathrm{if~} L=\neg E(c,b)\\
           t  & \mathrm{if~} L= E(b,a)\\
           0  & \mathrm{if~} L=\neg E(b,a)\\
           0  & \mathrm{for~the~other~positive~facts}\\
           1  & \mathrm{for~the~other~negative~facts.}
            \end{cases}
$
\end{minipage}\\

So, for example, $\beta(E(c,a))=\beta(E(c,c))=\ldots=0$ and also
$\beta(\neg E(c,a))=\beta(\neg E(c,c))=\ldots=1$. 
Note that $\beta$ is \ModDef\ in the sense of
Definition~\ref{def:int-model-defining}
and that the model it defines is precisely $\gothG$.

\begin{figure}
\centering
\begin{tikzpicture}
\begin{scope}
  \node[shape=circle,draw=black,very thick] (A) at (0,0) {$a$};
  \node[shape=circle,draw=black,very thick] (B) at (2,1.5) {$b$};
  \node[shape=circle,draw=black,very thick] (C) at (4,0) {$c$};

 \draw[->,very thick] 
      (A) edge[bend right=20] node[right] {$\,p$} (B);
 \draw[->,very thick] 
      (B) edge[bend right=20] node[left] {$q\,$} (C);
 \draw[->,very thick,dashed] (A) edge node[below] {$\nnr$} (C);
 \draw[->,very thick,dashed] 
      (C) edge[bend right=20] node[right] {$\,\nns$} (B);
 \draw[->,very thick] 
      (B) edge[bend right=20] node[left] {$t\,$} (A);
\end{scope}
\end{tikzpicture}
\vspace*{-3mm}
\caption{The model $\gothG$}
\label{fig:modelG}
\end{figure}

The assumptions made in the definition of $\beta$ 
indicate that we choose to track positive facts
like $E(b,c)$ and negative facts like $\neg E(a,b)$, etc., as they are
used in establishing the truth of some sentence in $\gothG$.
They also indicate that we accept, and thus do not
track, the absence of the other potential edges such as $E(c,a)$.
We think of data annotated with $0$ as being ``forget-about-it'' absent and of 
data annotated with 1 as ``available for free'' present.

Clearly, $\gothG\models\phi$, but how can we justify this in terms
of the facts, negative or positive, that hold in the model?
By computing the semantics of the sentence $\phi$
under the interpretation $\beta$ we will obtain \emph{provenance information}
for the result $\gothG\models\phi$. Clearly
$$
\nnf(\phi)~\equiv~ 
\forall x\,\exists y\: \bigl( x\neq y \wedge (\neg E(x,y)\vee E(y,x))\bigr)
$$
and therefore
$$
\beta\sem{\phi}{} ~=~ \beta\sem{\nnf(\phi)}{} ~=~
(\nnr+t)\cdot p\cdot(1+q+\nns) ~=~ p\nnr+pt+pq\nnr+pqt+p\nnr\nns+p\nns t.
$$

Each of the monomials of the dual polynomial $\beta\sem{\phi}{}$
has coefficient 1
\footnote{In this example all the monomial
coefficients and all the exponents are 1. This is certainly not the case
in general. In fact, it is possible to show that any dual polynomial
can be computed as some provenance, with suitable choices of sentence,
model, and interpretation.}
and each corresponds to a different (model-checking) 
proof tree of $\phi$ from the literals
described by the monomial.
For example, the monomial $pt$ corresponds to a proof tree of $\phi$ 
in which the fact
$E(a,b)$ is used to deny the dominance of $b$, the fact
$E(b,a)$ is used to deny the dominance of $a$, and the
negative fact $\neg E(c,a)$, which is accepted without 
tracking---it has provenance 1---is used to deny the dominance of $c$.

Note that what we call proof tree here involves formulae in NNF
and has inference rules corresponding
to model checking conjunction, disjunction, universal and existential
quantifiers. We illustrate with the proof tree 
corresponding to another monomial, $p\nnr\nns$,
using the following formula abbreviations:
\begin{eqnarray*}
\denydom(x,y) &\equiv& \bigl(x\neq y \wedge (\neg E(x,y)\vee E(y,x))\bigr)
~~~~~~~~~~~~y~\mathrm{denies~dominance~of}~x\\
\\
\notdom(x)   &\equiv& 
       \exists y\: \bigl( x\neq y \wedge (\neg E(x,y)\vee E(y,x))\bigr)
~~~~~~~~~~~x~\mathrm{is~not~dominant}\\
\\
\noVdom      &\equiv& 
     \forall x\,\exists y\: \bigl(x\neq y \wedge (\neg E(x,y)\vee E(y,x))\bigr)
~~~~~~~\mathrm{no~vertex~is~dominant}\\
\end{eqnarray*}

\vspace*{-5mm}
With these, the proof tree corresponding to $p\nnr\nns$ is:

\bigskip

\EnableBpAbbreviations

\noindent
\AXC{$a\neq b$}
\AXC{$\neg E(a,c)~~~[\nnr]$}
\UIC{$\neg E(a,c)\vee E(c,a)$}
\BIC{$\denydom(a,c)$}
\UIC{$\notdom(a)$}
\AXC{$b\neq c$}
\AXC{$E(a,b)~~~[p]$}
\UIC{$\neg E(b,a)\vee E(a,b)$}
\BIC{$\denydom(b,a)$}
\UIC{$\notdom(b)$}
\AXC{$c\neq a$}
\AXC{$\neg E(c,b)~~~[\nns]$}
\UIC{$\neg E(c,b)\vee E(b,c)$}
\BIC{$\denydom(c,b)$}
\UIC{$\notdom(c)$}
\TIC{$\noVdom$}
\DP\\

The following proposition summarizes the situation.

\begin{proposition}
\label{prop:all-proofs-mod-def}
Let $\beta:\LitA\rightarrow\poly{\Nat}{X,\nnX}$ be a provenance-tracking
\ModDef\ interpretation,
and let $\phi$ be an FOL sentence. 
Then, the dual polynomial
$\beta\sem{\phi}{}$ describes \emph{all} the proof
trees that verify $\phi$ using premises from among the literals that
that $\beta$ maps to provenance tokens or to 1 (i.e., from the
literals that hold in $\Abeta$).
Specifically, each monomial $m\,x_1^{m_1}\cdots x_k^{m_k}$ 
corresponds to $m$ distinct proof trees
that use $m_1$ times a literal that $\beta$ annotates by $x_1$, \ldots, 
and $m_k$ times a literal annotated by $x_k$, as well
as any number of the literals annotated with 1.
In particular, $\beta\sem{\phi}{}\neq0$ if, and only if, some proof tree exists, and if, and only if,
$\Abeta\models\phi$.
\end{proposition}

Note that since $\poly{\Nat}{X,\nnX}$ is not positive this proposition
does not follow from Proposition~\ref{prop:strong}.
(Nor does this contradict Proposition~\ref{prop:refi} (b) because 
provenance-tracking interpretations have a special form.)
Nonetheless, albeit not positive,
$\poly{\Nat}{X,\nnX}$ has many remarkable properties and
this proposition is a corollary of a more general one that we shall 
state in Sect.~\ref{subsec:prop-prov}.

\subsection{From Provenance to Confidence}
\label{subsec:conf}

Recall from Sect.~\ref{subsec:commutative-semirings} the Viterbi semiring
$\Vit$. We think of the elements of $\Vit$ as confidence scores. 
Going back to the example in Sect.~\ref{subsec:ex}, and assuming
specific confidence scores for the literals that $\gothG$ makes true,
and that we track, we wish to compute a confidence score for
$\gothG\models\phi$.

Specifically, consider the $\Vit$-interpretation 
$\gamma:\LitV\rightarrow[0,1]$ defined by
$$
\gamma(E(a,b))=\gamma(E(b,c))=0.9,
~~~~~~~~
\gamma(E(b,a))=0.2,
~~~~~~~~~~~~~
\gamma(\neg E(a,c))=\gamma(\neg E(c,b))=0.6,
$$
and in addition, for any \emph{other} positive fact we have
$\gamma(E(\_\,,\_))=0$ and for any \emph{other} negative fact 
we have $\gamma(\neg E(\_\,,\_))=1$.

With this we could use Definition~\ref{def:sem}
to compute $\gamma\sem{\phi}{}\in[0,1]$, which is the desired
confidence score.

However, since we have already computed in Sect.~\ref{subsec:ex} the provenance
$\beta\sem{\phi}{}$ we can take advantage of the Fundamental Property
(Proposition~\ref{prop:hom}) via a homomorphism whose existence is
guaranteed by Proposition~\ref{prop:prov-univ}.

We define $f:X\cup\nnX\rightarrow[0,1]$ by
$$
f(p)=f(q)=0.9,
~~~~~~~
f(t)=0.2,
~~~~~~~~~~~~~~~~
f(\nnr)=f(\nns)=0.6,
$$
by $f(x)=0$ for $x\not\in\{p,q,t\}$, and by $f(\nnx)=1$ for 
$\nnx\not\in\{\nnr,\nns\}$. The condition on $f$ in 
Proposition~\ref{prop:prov-univ} is satisfied, hence $f$ can be extended
to a homomorphism $h:\poly{\Nat}{X,\nnX}\rightarrow\Vit$. From the definition
of $f$ we have $h\circ\beta=\gamma$. By the Fundamental Property
$$
\gamma\sem{\phi}{}~=~ h(\beta\sem{\phi}{}).
$$
Hence the score we wish to compute can be obtained by applying the homomorphism
$h$ to the dual polynomial $\beta\sem{\phi}{}=p\nnr+pt+pq\nnr+pqt+p\nnr\nns+p\nns t$. It is easier to use the 
factored form of $\beta\sem{\phi}{}$:
$$
h(p(\nnr+t)(1+q+\nns)) ~=~ 0.9\cdot\max(0.6,\;0.2)\cdot\max(1,\;0.9,\;0.6)
~=~ 0.54.
$$
In general, confidence calculation may be only one of the analyses that we 
wish to perform. When these analyses are based on semiring calculations
we can compute the provenance just once and then evaluate it in multiple
semiring and under multiple valuations, by virtue of the Fundamental Property.

\subsection{Detailed Provenance Analysis: Top-Secret Proofs}
\label{subsec:clear}

We describe here another kind of provenance analysis that we can perform on
in conjunction with interpretation in various semiring.
Recall from Sect.~\ref{subsec:commutative-semirings} the access control
semiring $\Access$. Its elements are interpreted as \emph{clearance
  levels}, from lowest to highest
$\Pub<\Cnf<\Sec<\Tsec<0$. For example, administrators would assign 
clearance levels
to the different items in the input data. The resulting clearance
level for the output of a computation determines which users get to
access that output.  In the context of this paper there would be an
assignment of clearance levels to literals. 

Going back to the example in Sect.~\ref{subsec:ex}, 
consider the $\Access$-interpretation 
$\alpha:\LitV\rightarrow\Access$ defined by
$$
\alpha(E(a,b))=\alpha(E(b,c))=\alpha(E(b,a))=\Pub,
~~~~~~~~~~~~~
\alpha(\neg E(a,c))=\alpha(\neg E(c,b))=\Tsec,
$$
and in addition, for any \emph{other} positive fact we have
$\alpha(E(\_\,,\_))=0$ and for any \emph{other} negative fact 
we have $\alpha(\neg E(\_\,,\_))=\Pub$.

As in Sect.~\ref{subsec:conf} we have
$\alpha\sem{\phi}{}=h(p\nnr+pt+pq\nnr+pqt+p\nnr\nns+p\nns t)$,
where $h$ is the unique homomorphism $\poly{\Nat}{X,\nnX}\rightarrow\Access$
such that $h(p)=h(q)=h(t)=\Pub$,
$h(\nnr)=h(\nns)=\Tsec$, and otherwise equals $0$ on the rest of $X$
and equals $\Pub$ on the rest of $\nnX$. 

We can see that $\alpha\sem{\phi}{}=\Pub$ but we can also perform a
more detailed analysis in which we can associate clearance levels to
individual proof trees 
Thus, while it will be
publicly known that $\gothG\models\phi$, those with top-secret
clearance can also know that $p\nnr$ describes a proof of the
assertion $\gothG\models\phi$. This may become relevant if we have
particularly high confidence (as described above in Sect.~\ref{subsec:conf})
in the literals that $p$ and $\nnr$ annotate, that is, in the presence
of the edge from $a$ to $b$ and in the absence of an edge from $a$ to $c$.

\section{Reverse Provenance Analysis}
\label{sec:reverse}

There are limitations to what we can do with the provenance of a
model-checking assertion $\gothA\models\phi$ for a given $\gothA$.
It is even more interesting to consider provenance-tracking interpretations
that allow us to \emph{choose}, from among multiple models, the ones
that fulfill various desiderata.

\begin{figure}
\centering
\begin{minipage}[t]{.5\textwidth}
\centering
\begin{tikzpicture}
\begin{scope}
  \node[shape=circle,draw=black,very thick] (A) at (0,0) {$a$};
  \node[shape=circle,draw=black,very thick] (B) at (2,1.5) {$b$};
  \node[shape=circle,draw=black,very thick] (C) at (4,0) {$c$};

 \draw[->,very thick,dotted] 
      (A) edge[bend right=20] node[right] {$p,\nnp$} (B);
 \draw[->,very thick,dotted] 
      (B) edge[bend right=20] node[left] {$q,\nnq\,$} (C);
 \draw[->,very thick,dotted] (A) edge node[below] {$r,\nnr$} (C);
 \draw[->,very thick,dotted] 
      (C) edge[bend right=20] node[right] {$s,\nns$} (B);
 \draw[->,very thick,dotted] 
      (B) edge[bend right=20] node[left] {$t,\nnt$} (A);
\end{scope}
\end{tikzpicture}
\vspace*{-3mm}
\captionof{figure}{Provenance tracking assumptions}
\label{fig:assump}
\end{minipage}%
\begin{minipage}[t]{.5\textwidth}
\centering
\begin{tikzpicture}
\begin{scope}
  \node[shape=circle,draw=black,very thick] (A) at (0,0) {$a$};
  \node[shape=circle,draw=black,very thick] (B) at (2,1.5) {$b$};
  \node[shape=circle,draw=black,very thick] (C) at (4,0) {$c$};

 \draw[->,very thick] 
      (A) edge[bend right=20] node[right] {$\,p$} (B);
 \draw[->,very thick] 
      (B) edge[bend right=20] node[left] {$q\,$} (C);
 \draw[->,very thick] (A) edge node[below] {$r$} (C);
 \draw[->,very thick] 
      (C) edge[bend right=20] node[right] {$\,s$} (B);
 \draw[->,very thick] 
      (B) edge[bend right=20] node[left] {$t\,$} (A);
\end{scope}
\end{tikzpicture}
\vspace*{-3mm}
\captionof{figure}{The model $\gothF$}
\label{fig:modelF}
\end{minipage}%
\end{figure}

\subsection{A Reverse Analysis Example}
\label{subsec:rev-ex}

Let $V=\{a,b,c\}$ be a set of ground values. As before, these will eventually
play the role of the vertices of a digraph.
However, we do not yet specify
a set of edges, i.e., we do not specify a finite model with
universe $V$. Instead, as illustrated by the dotted edges in 
Figure~\ref{fig:assump}, we supply a set of provenance tokens 
$X=\{p,q,r,s,t\}$ that corresponds
to the \emph{potential presence} of some edges that we wish to track.
Therefore, $\nnX=\{\nnp,\nnq,\nnr,\nns,\nnt\}$ are the provenance tokens 
allowing us to track the \emph{potential absence} of the same edges.
These \textbf{provenance tracking assumptions} can be formalized via a 
provenance-tracking $\poly{\Nat}{X,\nnX}$-interpretation.

Define $\pi:\LitV\rightarrow X\cup\nnX\cup\{0,1\}$ by

\begin{minipage}[t]{7cm}
\hspace*{1cm}
$
\pi(L) ~=~ \begin{cases}
           p  & \mathrm{if~} L=E(a,b)\\
        \nnp  & \mathrm{if~} L=\neg E(a,b)\\
           q  & \mathrm{if~} L=E(b,c)\\
        \nnq  & \mathrm{if~} L=\neg E(b,c)\\
           r  & \mathrm{if~} L=E(a,c) \\
        \nnr  & \mathrm{if~} L=\neg E(a,c)
            \end{cases}
$
\end{minipage}
\begin{minipage}[t]{7cm}
$        =~\begin{cases}
           s  & \mathrm{if~} L= E(c,b)\\
        \nns  & \mathrm{if~} L=\neg E(c,b)\\
           t  & \mathrm{if~} L= E(b,a)\\
        \nnt  & \mathrm{if~} L=\neg E(b,a)\\
           0  & \mathrm{for~the~other~positive~facts}\\
           1  & \mathrm{for~the~other~negative~facts.}
            \end{cases}
$
\end{minipage}\\

So, for example, $\pi(E(c,a))=\pi(E(c,c))=\ldots=0$ and also
$\pi(\neg E(c,a))=\pi(\neg E(c,c))=\ldots=1$. 
This particular interpretation does not feature a positive fact annotated
with 1 but we could have just as well had $\pi(E(a,b))=1$ and 
$\pi(\neg E(a,b))=0$ if we chose to assume that edge without tracking it.

Note that $\pi$ is not \ModDef\ (in the sense of
Definition~\ref{def:int-model-defining}), i.e., it does not correspond
to any \emph{single} model. 
As we shall see, this is not a bug but a feature (!), as it
will allow us to consider, under the given provenance assumptions,
multiple models that can satisfy a sentence.

Now we compute the semantics of the sentence $\phi$ from Sect.~\ref{subsec:ex},
under this interpretation and we obtain 
$$
\pi\sem{\phi}{} ~=~ 
(\nnp+\nnr+t)\cdot(p+\nnq+s+\nnt)\cdot(1+q+r+\nns).
$$

If we multiply these three expressions and we apply
$p\nnp=q\nnq=r\nnr=s\nns=0$ we get a polynomial with 
$48-4-3-3-4=34$ monomials (the reader shall be spared the
trouble of admiring it). As in Sect.~\ref{subsec:ex},
each of these monomials has coefficient 1
and (as shown in Sect.~\ref{subsec:prop-prov})
each corresponds to a different proof tree of $\phi$ from the literals
described by the monomial.

For example, the monomial $pqt$ corresponds to a proof tree of $\phi$ 
in which the fact
$E(b,a)$ is used to deny the dominance of $a$, the fact
$E(a,b)$ is used to deny the dominance of $b$, and the fact
$E(b,c)$ is used to deny the dominance of $c$. Recalling the
notations from Sect.~\ref{subsec:ex}, note that the same monomial
is part of the dual polynomial $\beta\sem{\phi}{}$ and that
the same proof tree justifies $\gothG\models\phi$. Note also
that setting $r=s=\nnp=\nnq=\nnt=0$ in the definition
of $\pi$ gives the definition of $\beta$. Doing the same in 
$\pi\sem{\phi}{}$ gives
$$
(0+\nnr+t)\cdot(p+0+0+0)\cdot(1+q+0+\nns) 
~=~
(\nnr+t)\cdot p\cdot(1+q+\nns),
$$ 
which is the same as the polynomial $\beta\sem{\phi}{}$ obtained with the
\ModDef\ interpretation $\beta$ which corresponds to the model $\gothG$. In
this sense, $\pi$ is a ``generalization'' of $\beta$, or, 
$\beta$ can be obtained by \emph{specializing} $\pi$.
All this will be made precise in full generality
in Sect.~\ref{subsec:prop-prov} while here we explore two other interesting
specializations of $\pi$.

One of the monomials in $\pi\sem{\phi}{}$ is $\nnp\nnq$.  This means
that we can find a specialization of $\pi$ that is \ModDef\ and that defines,
in fact, 
a model with \emph{no} positive information, namely the digraph with
vertices $V$ and no edges.
Hence, denoting with $\gothE$
this no-edge model, we have $\gothE\models\phi$.
How many proof trees verify that $\gothE\models\phi$?
The specialization $\pi_1$ that we are after corresponds to
setting $p=q=r=s=t=0$. This gives
$$
\pi_1\sem{\phi}{} ~=~ 
(\nnp+\nnr)\cdot(\nnq+\nnt)\cdot(1+\nns),
$$
which is a polynomial with 8 monomials, each with coefficient 1. It follows
that there are 8 distinct proof trees for $\gothE\models\phi$.

One can also figure out that $pqt, prt, qst, rst$ are among the
monomials in $\pi\sem{\phi}{}$. This means that we can find another
specialization of $\pi$ that is also \ModDef\ and that defines a model with
\emph{maximum} positive information (allowed by $\pi$), namely the 
digraph with
vertices $V$ and edges $E(a,b),E(b,c),E(a,c),E(c,b)$ and $E(b,a)$.
Let's denote with $\gothF$
this all-allowed-edges model (see Figure~\ref{fig:modelF}). 
How many proof trees verify that $\gothF\models\phi$?
The specialization $\pi_2$ that we look for here 
corresponds to setting $\nnp=\nnq=\nnr=\nns=\nnt=0$. This gives
$$
\pi_2\sem{\phi}{} ~=~ 
t\cdot(p+s)\cdot(1+q+r),
$$
which is a polynomial with 6 monomials, each with coefficient 1, 
hence there are 6 proof trees
for this.

Finally, we also wish to consider for this example 
the provenance of the \emph{negation}
of the sentence $\phi$ considered above, i.e., the sentence
$\neg\phi$ that says that the digraph \emph{has} a dominant vertex:
$$
\neg \phi~\equiv~ \neg \forall x\,\neg\domi(x).
$$
Since $\domi(x) \equiv 
\forall y\: \bigl(x=y \vee (E(x,y)\wedge\neg E(y,x))\bigr)
$
is already in NNF, we have $\nnf(\neg\phi)\equiv\exists x\,\domi(x)$.
We compute the semantics of this sentence under the same interpretation:
$$
\pi\sem{\neg\phi}{}~=~ pr\nnt + \nnp q\nns t + s\nnq\nnr\cdot0
                   ~=~ pr\nnt + \nnp q\nns t.
$$
Thus, under the provenance tracking assumptions we have made,
there are only two proof trees for $\neg\phi$.
\footnote{In the polynomials featured in this example all the monomial
coefficients and all the exponents are 1. This is certainly not the case
in general. In fact, it can be shown that any dual polynomial results
from suitably chosen sentences and interpretations. 
We shall come back later to this.}


\subsection{Properties of Provenance} 
\label{subsec:prop-prov}

In this subsection all interpretations are provenance-tracking, unless
another semiring is specified.  The interpretation exhibited
in Sect.~\ref{subsec:rev-ex} belongs to a class that merits its own
definition.

\begin{definition} A provenance-tracking interpretation
$\pi:\LitA\rightarrow\poly{\Nat}{X,\nnX}$
is said to be \bfModComp\
if for each fact $R(\bara)$ one of the
following three holds:
\begin{enumerate}
\item $\exists z\in X$
s.t $\pi(R(\bara))=z$ and $\pi(\neg R(\bara))=\nnz$,
or
\item $\pi(R(\bara))=0$ and $\pi(\neg R(\bara))=1$, or
\item $\pi(R(\bara))=1$ and $\pi(\neg R(\bara))=0$
\end{enumerate}
\end{definition}

As promised, we state a more powerful
version of Proposition~\ref{prop:all-proofs-mod-def}
(which was about provenance-tracking \ModDef\ interpretations).

\begin{proposition}
\label{prop:all-proofs}
Let $\pi:\LitA\rightarrow\poly{\Nat}{X,\nnX}$ be a \ModComp\
interpretation and let $\phi$ be an FOL sentence. Then,
$\pi\sem{\phi}{}$ describes \emph{all} the proof
trees that verify $\phi$ using premises from among the literals that
$\pi$ maps to provenance tokens or to 1.  Specifically, each monomial
$m\,x_1^{m_1}\cdots x_k^{m_k}$ corresponds to $m$ distinct proof trees
that use $m_1$ times a literal annotated by $x_1$, \ldots, and $m_k$ times
a literal annotated by $x_k$., where $x_1,\ldots,x_k\in X\cup\nnX$.
In particular, when $\pi\sem{\phi}{}=0$ no proof tree exists.
\end{proposition}

\begin{corollary}
\label{cor:count}
Let $\pi$ be a \ModComp\ interpretation. Then, the sum of
the monomial coefficients in $\pi\sem{\phi}{}$ counts 
the number of proof trees that verify $\phi$ using premises 
from among the literals that $\pi$ maps to provenance tokens or to 1.
The same count can be obtained from an $\Nat$-interpretation 
as $(h\circ\pi)\sem{\phi}{}\in\Nat$ where 
$h:X\cup\nnX\cup\{0,1\}\rightarrow\Nat$ is defined by $h(0)=0$ and
$h(p)=h(\nnp)= h(1)=1$.
\end{corollary}

A \ModComp\ interpretation may allow the tracking of both a literal
and its negation. Therefore, \ModComp\ interpretations are not 
\ModDef\ unless they do not make use of provenance tokens at all 
(in which case they are essentially canonical truth interpretations). 
Hence, 
Proposition~\ref{prop:all-proofs-mod-def} is not a simple particular case of 
Proposition~\ref{prop:all-proofs}. Nonetheless, we shall see how
\ModDef\ interpretations can be seen as \emph{specializations} of \ModComp\ 
interpretations with respect to models that ``agree'' 
(i.e., are \emph{compatible}) with them, as defined below. 

\begin{definition}
Let $\pi:\LitA\rightarrow\poly{\Nat}{X,\nnX}$ be a \ModComp\
interpretation and let $\gothA$
be a model with universe $A$ (same $A$). We say that $\gothA$
is \textbf{compatible} with $\pi$
if $\gothA\models L$ for any literal $L$ such that $\pi(L)=1$. 
Further, let $\Modpi:=\{\gothA\mid\gothA\mathrm{~is~compatible~with~}\pi\}$.
\end{definition}

For instance, the models shown in 
Figures~\ref{fig:modelG} and~\ref{fig:modelF}
are compatible with the interpretation
defined in Sect.~\ref{subsec:rev-ex}.

Now we can talk about satisfiability and validity \emph{restricted to the
class of models that agree with the provenance tracking assumptions}
made by an interpretation.

\begin{corollary}[to Proposition~\ref{prop:all-proofs}]
\label{cor:sat}
Let $\pi:\LitA\rightarrow\poly{\Nat}{X,\nnX}$ be a \ModComp\
interpretation and let $\phi$ be a first-order sentence. 
Then, $\phi$ is $\Modpi$-satisfiable if, and only if,  $\pi\sem{\phi}{}\neq0$,
and $\phi$ is $\Modpi$-valid if, and only if,  $\pi\sem{\neg\phi}{}=0$.
\end{corollary}

This is not finite satisfiability (shown undecidable by Trakhtenbrot),
of course.
Even if we map every possible literal to a different provenance token
we only decide satisfiability in a model with exactly $|A|$ elements,
which is easily in NP (without talking about provenance).

\begin{example}
\label{ex:tautology}
With the same (digraph) vocabulary as in Sect.~\ref{subsec:ex} 
and~\ref{subsec:rev-ex} consider the sentence
$$
\tau := \exists x\,\forall y\,E(x,y)\,\rightarrow\,
            \forall y\,\exists x\,E(x,y).
$$
This is a well-known tautology (holding in all models, not just in finite
ones). Obviously,
$\nnf(\neg\tau) = \exists x\,\forall y\,E(x,y))\,\wedge\,\exists y\,\forall x\,\neg E(x,y)$. 
Now consider $V=\{a,b\}$ and a truth-compatible interpretation
$\pi$ that annotates $E(a,b), E(b,a), E(a,a), E(b,b)$ with $p,q,r,s$ 
respectively,
and the corresponding negated facts with $\nnp,\nnq,\nnr,\nns$. Then
$$
\pi\sem{\neg\tau}{}~=~(pr+qs)(\nnq\nnr+\nnp\nns)=0,
$$
verifying that $\tau$ is $\Modpi$-valid. 
\end{example}

From the provenance analysis of (provenance-restricted) 
validity/satisfiability that is enabled by Corollary~\ref{cor:sat}
we can obtain a provenance analysis of model checking, for each
model of a given sentence, as follows.

\begin{definition}
Let $\pi$ be \ModComp\ and let $\gothA\in\Modpi$.
The \textbf{specialization} of $\pi$ with respect to $\gothA$
is the $\poly{\Nat}{X,\nnX}$-interpretation
$\refi{\pi}{\gothA}:\LitA\rightarrow\poly{\Nat}{X,\nnX}$ defined by
$$
\refi{\pi}{\gothA}(L) ~=~ \begin{cases}
                          \pi(L) & \mathrm{if~}\gothA\models L\\
                          0      & \mathrm{otherwise.}
                          \end{cases}
$$
\end{definition}

Note that $\refi{\pi}{\gothA}$ is always \ModDef\ and the model it defines
is, of course, $\gothA$.

The \ModDef\ interpretation $\beta$ in Sect.~\ref{subsec:ex} is the specialization
with respect to the model $\gothG$ of the \ModComp\ interpretation $\pi$
in Sect.~\ref{subsec:rev-ex}, $\beta=\refi{\pi}{\gothG}$. 
Other specializations of $\pi$ are given
in Sect. ~\ref{subsec:rev-ex}. The next corollary finally justifies 
Proposition~\ref{prop:all-proofs-mod-def}.

\begin{corollary}[to Proposition~\ref{prop:all-proofs}]
\label{cor:model-check}
Let $\pi:\LitA\rightarrow\poly{\Nat}{X,\nnX}$ be a \ModComp\
interpretation, let $\gothA$ be structure that is compatible with $\pi$, and let 
$\phi$ be a first-order sentence such that $\gothA\models\phi$ (hence, by
Corollary~\ref{cor:sat}, $\pi\sem{\phi}{}\neq0$).

Then, $\refi{\pi}{\gothA}\sem{\phi}{}\neq0$ and
every monomial in $\refi{\pi}{\gothA}\sem{\phi}{}$
also appears in $\pi\sem{\phi}{}$, with the same coefficient.

Moreover, $\refi{\pi}{\gothA}\sem{\phi}{}$ describes 
\emph{all} the proof trees that verify $\gothA\models\phi$. In particular,
the sum of all the monomial coefficients in $\refi{\pi}{\gothA}\sem{\phi}{}$
counts the number of distinct such proof trees 
(as in Corollary~\ref{cor:count}, the same count can be obtained from an
$\Nat$-interpretation).
\end{corollary}

While $\refi{\pi}{\gothA}\sem{\phi}{}$ analyzes the provenance
of checking in a specific model, the more general
$\pi\sem{\phi}{}$ allows for a form \emph{reverse analysis}.
Indeed, to each monomial $m\,x_1^{m_1}\cdots x_k^{m_k}$ in
$\pi\sem{\phi}{}\neq0$ we can associate a model from $\Modpi$ that
makes true the literals that are annotated by $x_1,\ldots,x_k$
(and possibly more literals)
and, as we have seen, every model $\gothA\in\Modpi$
such that $\gothA\models\phi$ can be obtained this way.

\begin{example}[Example~\ref{ex:tautology} cont'd]
\label{ex:tautology-two}
Let us also compute the provenance of the tautology $\tau$ itself:
$$
\pi\sem{\tau}{}=(\nnp+\nnr)(\nnq+\nns)+(q+r)(p+s).
$$
Here $\Modpi$ consists of all possible structures with universe $\{a,b\}$
and, for any such $\gothA$, the model-refinement $\refi{\pi}{\gothA}$
sets to 0 exactly one of the two tokens in a complementary pair.
No matter how this is done, observe that 
$\refi{\pi}{\gothA}\sem{\tau}{}\neq0$.
\end{example}

\subsection{Confidence Maximization} 
\label{subsec:conf-max}

\begin{figure}
\centering
\begin{tikzpicture}
\begin{scope}
  \node[shape=circle,draw=black,very thick] (A) at (0,0) {$a$};
  \node[shape=circle,draw=black,very thick] (B) at (2,1.5) {$b$};
  \node[shape=circle,draw=black,very thick] (C) at (4,0) {$c$};

 \draw[->,very thick] 
      (A) edge[bend right=20] node[right] {$\,1/3$} (B);
 \draw[->,very thick,dotted] 
      (B) edge[bend right=20] node[left] {$q,\nnq\,$} (C);
 \draw[->,very thick] (A) edge node[below] {$1/3$} (C);
 \draw[->,very thick,dotted] 
      (C) edge[bend right=20] node[right] {$\,s,\nns$} (B);
 \draw[->,very thick,dashed] 
      (B) edge[bend right=20] node[left] {$1/3\,$} (A);
\end{scope}
\end{tikzpicture}
\vspace*{-3mm}
\caption{Maximum confidence model with dominant vertex}
\label{fig:max-conf}
\end{figure}
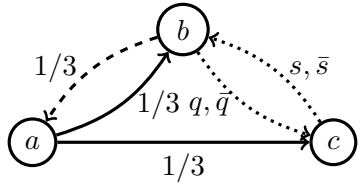

As in Sect.~\ref{subsec:conf} we use the Viterbi semiring $\Vit$
from Sect.~\ref{subsec:commutative-semirings} interpreting its values as
confidence scores.  
Interestingly, we can reverse analyze the provenance polynomials 
and use confidence scores to find a model in which confidence is
maximized.

In the context of the example in Sect.~\ref{subsec:rev-ex},
suppose that we have confidence $1/3$ in all the literals
that the \ModComp\ interpretation $\pi$ maps to
a (positive or negative) provenance token. This yields 
a $\Vit$-interpretation $\pi'$ which, by Propositions~\ref{prop:hom}
and~\ref{prop:prov-univ},
factors as $\pi'= h\circ\pi$  where 
$h$ is the unique semiring homomorphism
$\poly{\Nat}{X\cup\nnX}\rightarrow\Vit$ that maps all the tokens
$p,\ldots,\nnp,\ldots$ to $1/3$ (this is perfectly plausible, as 
confidence is \emph{not} probability).

Now recall from Sect.~\ref{subsec:rev-ex}
the sentence $\neg\phi$ (which asserts that there exists a dominant
vertex).  We have computed $\pi\sem{\neg\phi}{} ~=~ pr\nnt + \nnp
q\nns t$. Obviously, $\pi'$ is inconsistent so further applying
$h(pr\nnt + \nnp q\nns t)= 1/27+1/81= 4/81$ is not
meaningful. However, we know from Corollary~\ref{cor:model-check} that each
monomial in $\pi\sem{\neg\phi}{}$ corresponds to some model of $\neg\phi$. In
this case we have exactly two proof tree choices, corresponding to different
models, and they give different
confidence to $\neg\phi$. To maximize confidence we choose the
monomial $pr\nnt$ therefore a model in which we have an edge $E(a,b)$,
an edge $E(a,c)$ and \emph{no} edge $E(b,a)$.  This will ensure the
  dominance of vertex $a$ with confidence $1/27$, in other words,
$\neg\phi$ is $1/27$-true in this model.  
This model is shown in Figure~\ref{fig:max-conf} (the
 edge $E(b,a)$ is dashed because it is absent but we still wanted to show
  the confidence 1/3 in this absence). 
The edges $E(b,c)$ and $E(c,b)$ are dotted because neither their presence
nor their absence contradicts the provenance assumptions. We can, in fact, 
continue with a provenance analysis for these two edges if other properties
of the model are of interest.

\section{Model Update}

\begin{figure}
\centering
\begin{tikzpicture}
\begin{scope}
  \node[shape=circle,draw=black,very thick] (A) at (0,0) {$a$};
  \node[shape=circle,draw=black,very thick] (B) at (2,1.5) {$b$};
  \node[shape=circle,draw=black,very thick] (C) at (4,0) {$c$};

 \draw[->,very thick,dashed] 
      (A) edge[bend right=20] node[right] {$\,\nnp$} (B);
 \draw[->,very thick,dashed] 
      (B) edge[bend right=20] node[left] {$\nnq\,$} (C);
 \draw[->,very thick,dashed] (A) edge node[below] {$\nnr$} (C);
 \draw[->,very thick,dashed] 
      (C) edge[bend right=20] node[right] {$\,\nns$} (B);
 \draw[->,very thick] 
      (B) edge[bend right=20] node[left] {$t\,$} (A);
\end{scope}
\end{tikzpicture}
\vspace*{-3mm}
\caption{The model $\gothH$}
\label{fig:modelH}
\end{figure}

In this section we indicate a method for updating provenance polynomials
corresponding to a \ModDef\ interpretation when the model associated with
the interpretation is updated by inserting or deleting facts.

For example, recall from Sect.~\ref{subsec:ex} the interpretation $\beta$, 
the structure $\gothG$ that it defines (Figure~\ref{fig:modelG}), and the sentence
$\phi$ asserting ``no dominant vertex''. We had computed
$$
\beta\sem{\phi}{} ~=~ (\nnr+t)\cdot p\cdot(1+q+\nns).
$$
First suppose that we update $\gothG$ by \emph{deleting} $E(a,b)$ and $E(b,c)$. 
Keeping the other provenance targets, this results in the model 
$\gothH$ depicted in Figure~\ref{fig:modelH}. 
What is the 
corresponding update on the dual polynomial $\beta\sem{\phi}{}$? For the
provenance polynomials used for positive queries, as in~\cite{pods/GreenKT07},
this update is performed by setting $p=q=0$. However, this would result
in the polynomial 0, which is wrong, because $\gothH\models\phi$.

The right way to perform this update takes advantage of the results 
in Sect.~\ref{subsec:prop-prov}. We use the \ModComp\ interpretation $\pi$
given in Sect.~\ref{subsec:rev-ex} (or any other \ModComp\ interpretation
that both $\gothG$ and $\gothH$ are compatible with and that specializes
with respect to $\gothG$ to $\beta$). Recall from Sect.~\ref{subsec:rev-ex} that
$$
\pi\sem{\phi}{} ~=~ 
(\nnp+\nnr+t)\cdot(p+\nnq+s+\nnt)\cdot(1+q+r+\nns),
$$
and therefore
$$
\refi{\pi}{\gothH}\sem{\phi}{} ~=~ 
(\nnp+\nnr+t)\cdot\nnq\cdot(1+\nns)
$$
is the update we desire. Comparing this with $\beta\sem{\phi}{}$
shows the need for doing an excursion through $\pi$.

Next, suppose that we update $\gothG$ by \emph{inserting} $E(a,c)$ and $E(c,b)$
resulting in the model $\gothF$ in Figure~\ref{fig:modelF}. Then, the update
of $\beta\sem{\phi}{}$ is
$$
\refi{\pi}{\gothF}\sem{\phi}{} ~=~ 
t\cdot(p+s)\cdot(1+q+r).
$$

\section{Conclusions}

The previous work on provenance in databases focused on positive languages,
and essentially even on the $\exists,\wedge,\vee$ fragment of first-order logic. But it also
focused on Datalog, hence on least fixed points. The presentation in
this article should encourage us to extend these studies to the full least fixed-point logic
LFP. This will be done in subsequent work, in relationship to games. The model checking
games for LFP are parity games (see e.g. \cite{AptGra11}), which are much more complicated than the acyclic 
games with only finite plays that suffice for first-order logic. At this point it is not really clear yet how
a provenance analysis for arbitrary parity games can be done, but  it is known that, on finite structures, we can
restrict LFP to formulae that only make use of positive least fixed-point operators, without losing expressive power.
On the game-theoretic side this corresponds to restricting parity games to reachability games (that however may 
still admit infinite plays), and for these a combination of $\omega$-continuous semirings of formal power series
with the idea of dual indeterminates provides a sound mathematical basis for provenance analysis.

\paragraph{Acknowledgements}
Our collaboration on the topics of this paper started in Fall 2016 as we were
both participating in the ``Logical Structures in Computation''
program at the Simons Institute for the Theory of Computing in
Berkeley. We are very grateful to the Institute for support and for the
perfect collaborative atmosphere that it fosters. 
We would like to acknowledge very useful discussions at the Institute
with Andreas Blass, Miko{\l}aj Boja\'nczyk, 
Thomas Colcombet, Anuj Dawar, Kousha Etessami, Diego Figueira, 
Phokion Kolaitis, Ugo Montanari, Jaroslav Ne{\v s}et{\v r}il, 
Daniela Petri{\c s}an, and Miguel Romero.

Val Tannen is very grateful to his collaborators
in the development over several years of semiring
provenance for databases: (in chronological$\,\cdot\,$alphabetical
order) T.J.~Green, Grigoris Karvounarakis, Zack Ives, Nate Foster,
Yael Amsterdamer, Daniel Deutch, Tova Milo, Susan Davidson, Julia
Stoyanovich, Sudeepa Roy, and Yuval Moskovitch. He was partially supported
by NSF grants 1302212 and 1547360 and by NIH grant U01EB02095401.


\end{document}